\pdfoutput=1
\documentclass[pdflatex,sn-apa]{sn-jnl}

\usepackage{graphicx}%
\usepackage{multirow}%
\usepackage{amsmath,amssymb,amsfonts}%
\usepackage{amsthm}%
\usepackage{xcolor}%
\usepackage{textcomp}%
\usepackage{manyfoot}%
\usepackage{booktabs}%
\usepackage{algorithm}%
\usepackage{algorithmicx}%
\usepackage{algpseudocode}%
\usepackage{listings}%
\usepackage{tikz}
\usetikzlibrary{positioning, arrows.meta, calc}
\usepackage{float}
\usepackage{array}
\usepackage{geometry}
\DeclareUnicodeCharacter{00A0}{ }

\newtheorem{proposition}{Proposition}
\newtheorem{corollary}{Corollary}

\newtheorem{definition}{Definition}
\usetikzlibrary{shapes.geometric}

\begin{document}

\title[No Certificate, No Categorical Speech Act]{No Certificate, No Categorical Speech Act: A Brouwerian Assertibility Constraint for Public Machine Speech}

\author*[1]{\fnm{Michael} \sur{Jülich}}\email{juelich@aihorizon-rd.org}

\affil*[1]{\orgname{aihorizon R\&D}, \orgaddress{\city{Munich}, \country{Germany}}}

\abstract{Suppose a citizen asks a reliable AI chatbot whether a cabinet secretary engaged in serious corruption. What entitles the chatbot to answer yes or no? Drawing on Brouwer's account of constructive assertion, I argue that categorical machine speech in high-stakes domains requires an auditor-checkable and publicly contestable certificate of entitlement. Citizens expected to rely on a verdict must be able to identify and challenge the grounds on which it is issued. This requirement yields three statuses -- licensed assertion, licensed denial, and \textit{Undetermined} -- that mark entitlement to categorical speech, not truth. \textit{Undetermined} is neither a truth-value gap nor a tunable reject option. A design lemma shows what eliminating the third status would cost. Under stated assumptions, an interface that certifies a verdict in every case by establishing or refuting its deployed predicate provides a witnessed decision procedure on its declared scope. Where no such procedure is available, a certificate-sound interface cannot guarantee a categorical verdict for every case. I defend an output contract that keeps unsettled cases within scope and requires \textit{Undetermined} whenever the available certificates fail to license exactly one side. In the corruption case, the citizen can accept a committee's finding that payments bought favorable lease terms and still ask why the materiality threshold is set at \$500,000. The contract makes that choice explicit, and the certificate records how it is applied. Responsibility for defending the threshold remains with the institutions that set it, not with the machine.}


\maketitle

\section{Introduction}\label{sec:intro}
 
\begin{quote}
``Nothing that you will learn in the course of your studies will be of the slightest possible use to you in after life, save only this, that if you work hard and intelligently you should be able to detect when a man is talking rot, and that, in my view, is the main, if not the sole, purpose of education'' \citep[Smith, as cited in][p.~27]{horne1988macmillan}.
\end{quote}
 
Delivered in a 1914 speech at Oxford University, John Alexander Smith's remark now bears on detecting ``rot'' in machines that speak with the semblance of authority. Recent experiments suggest that large language models can shift voters' political preferences chiefly by supplying factual and evidential claims, some of them incorrect \citep{lin2025persuading}. In a separate study, the methods that made models more persuasive also made their claims less accurate \citep{hackenburg2025levers}.
 
These findings are commonly read as problems of accuracy. Yet even a system that is usually right would leave open what entitles it to speak categorically. I locate the danger in machine-generated claims presented as settled without grounds that citizens can identify and challenge. Such claims may bypass the capacities on which democratic epistemic agency depends \citep{coeckelbergh2023democracy}. Accordingly, I ask what justificatory standards AI systems must meet in high-stakes domains\footnote{``High-stakes'' denotes domains in which outputs can affect public deliberation, legal status, health, liberty, or comparable basic interests. The designation is functional and applies whenever outputs are foreseeably used in such decisions or in their justificatory context. It is fixed by the declared scope and deployment contract; where plausible harms to basic interests remain, prudence requires treating the domain as high-stakes.} if their outputs are to remain rationally acceptable. I answer with a design constraint on machine speech: an AI system may assert or deny a claim only when it can present an auditor-checkable and publicly contestable certificate of entitlement; otherwise it must abstain from categorical outputs and instead return \textit{Undetermined}. I call this the \textit{Brouwerian assertibility constraint}. Uncertainty quantification can contribute to certificates, but it does not by itself distinguish warranted assertion from unwarranted fluency. Nor does process reliability alone supply the exhibited and contestable ground that, I argue, categorical machine speech in such domains requires. My principle is simple: \textit{no certificate, no categorical speech act}.
 
The constraint makes \textit{Undetermined} more than a cautious option. An interface that answered every case in its scope with a certified verdict, either establishing the predicate it is deployed to decide or genuinely refuting it, would thereby supply a witnessed decision procedure for that predicate. This is the paper's \textit{design lemma}. Yet nothing guarantees that such witnesses are available for every case. Near decision thresholds and in contested public matters, the grounds at hand may settle neither side. There, something has to give: verdicts lose their warrant, denial collapses into mere failure to assert, or the hard cases are quietly removed from the scope. I argue that the system should instead preserve its declared scope and return \textit{Undetermined} whenever the available certificates do not license exactly one side. Read alongside Wright's \citeyearpar[Chapter~2]{wright1992truth} critique of Putnam's \citeyearpar[pp.~49--50]{putnam1981reason} epistemic conception of truth, \textit{Undetermined} emerges as an interface-level counterpart to intuitionistic resistance to unwarranted completeness. It restrains categorical machine speech and leaves the logic of truth unrevised.
 
That a machine can produce a sentence is not in itself a reason for anyone to stand behind it. Brouwer pressed a related distinction against formalism, insisting that rule-governed manipulation of signs cannot substitute for the construction that warrants an assertion. He also feared that such sign-manipulation could come to regulate social life \citep[p.~ix]{vandalen2026intuitionistic}, a concern that fluent machine verdicts make newly concrete. I therefore treat his demand for exhibited grounds not as a foundation for machine truth but as a normative limit on automated categorical speech. Certificates make entitlement inspectable without, in Cavell's sense, acknowledging what is said \citep[pp.~238--266]{cavell2002knowing} or taking responsibility for it on anyone's behalf. The responsibility for standing behind words remains ours.
 
Hence, the response must be educational as well as technical. It requires a shared practice of learning when to acknowledge claims and when to refuse them in contexts where authority is easy to simulate and difficult to justify. What follows articulates an algorithmic grammar of justification, translating the philosophical argument into a normative specification for machine speech.
 
Section~\ref{sec:brouwerian_toolkit} introduces the Brouwerian toolkit and the exhibited-witness conditions it imposes on categorical machine speech in high-stakes domains. Section~\ref{sec:interface_semantics} develops the three-status interface semantics and links internal witnesses to public standing through the certificate as a boundary object. Section~\ref{sec:tooth_social} uses an illustrative political case to show how the decomposition of a claim and the standard itself become open to challenge alongside the evidence. Section~\ref{sec:contract_lemma} states the output contract and the design lemma. Section~\ref{sec:limitations} addresses the framework's limitations, and Section~\ref{sec:conclusion} concludes.

\section{A Brouwerian Toolkit for Constructive Assertibility}\label{sec:brouwerian_toolkit}
 
Brouwer's mathematical intuitionism is well suited to the present argument because it unites construction, logic, and temporality within a single methodological tradition.
 
\subsection{Construction, Logic, and Temporal Warrant}\label{subsec:four_elements}
 
The intuitionistic foundation of mathematics aims not at formal security but at conceptual evidence. On this account, mathematics is a practice of construction through temporally ordered mental acts. A statement of the form $\exists x\,P(x)$ is justified only by a construction that produces an object satisfying $P$. Correspondingly, $\lnot\exists x\,P(x)$ is justified only by a construction showing that any purported witness would lead to contradiction. Logic enters only afterward, as a description of patterns within constructive activity, and not as an independent source of mathematical truth \citep[pp.~1--2, 13]{atten2004brouwer}.
 
If mathematics consists of constructions, logical form must likewise be interpreted constructively. The \textit{Brouwer--Heyting--Kolmogorov interpretation} (BHK) captures this idea by defining a proposition's meaning in terms of what would count as a proof of it. Arend Heyting puts the point linguistically by distinguishing a statement from its assertion: a mathematical statement expresses an intention, while its assertion reports that the intention has been fulfilled \citep[p.~113]{heyting1931intuitionistische}. The logical constants render this discipline operational. A conjunction requires proofs of both conjuncts; a disjunction requires a proof of one side together with an indication of which; an implication requires a construction that transforms proofs of the antecedent into proofs of the consequent; a negation requires a construction that transforms any proof of the proposition into a contradiction; and an existential claim requires an exhibited witness (the Supplementary Appendix, hereafter SA, gives the full schema in Section~C.1). Assertion is therefore licensed only when the constructive task expressed by the proposition has been successfully carried out.
 
Brouwer rejects the unrestricted use of the \textit{law of excluded middle} (LEM), especially when principles valid for finite cases are extended to infinite domains without constructive warrant \citep[p.~149]{ferreiros2008crisis}. Schematically, LEM holds that every proposition $p$ satisfies $(p\vee\lnot p)$, independently of whether either side can be demonstrated. For statements about infinite structures, Brouwer denies that such a disjunction has intuitive evidence in the absence of a construction establishing one of its terms.
 
Unsolved problems make the objection vivid. Consider the claim that the digit 7 appears infinitely often in the decimal expansion of $\pi$. Classically, either it does or it does not, even though neither alternative has been established. For Brouwer, asserting the disjunction without a construction illegitimately treats an undecided problem as though it were already decided \citep[p.~150]{ferreiros2008crisis}.
 
This resistance to unwarranted applications of excluded middle is inseparable from Brouwer's temporal conception of mathematical truth. He writes that ``[t]ruth is only in reality, i.e., in the present and past experiences of consciousness [\ldots]. There are no non-experienced truths'' \citep[p.~1243]{brouwer1949consciousness}. A mathematical proof is experienced through the successful execution of a construction, whether carried out directly or grasped as a method capable of completion. Proof is therefore a temporally situated sequence of acts. A proposition acquires a truth value when a proof or refutation is given; until then, it is not false but remains without a truth value \citep[p.~10]{atten2004brouwer}. Brouwer gives this temporal epistemology greater precision through a fourfold distinction among statuses \citep[p.~114]{brouwer1955effect}: (1) $p$ has been proved true; (2) $p$ has been proved false, i.e., absurd; (3) $p$ has not been proved true or absurd, but we know an algorithm that will eventually decide the matter one way or the other; and (4) $p$ has not been proved true or absurd, and we know of no such decision procedure. These distinctions mark the constructive status of a proposition at a given stage and introduce no third truth value.
 
\medskip
The three Brouwerian elements jointly supply a constructive discipline for public-facing machine speech. Construction and BHK semantics motivate the demand for exhibited entitlement with a compositional structure. The rejection of unrestricted excluded middle motivates withholding categorical output where neither side has been constructively secured. Finally, the temporal conception of truth makes entitlement responsive to the stage at which grounds have actually been exhibited. The point is not to attribute Brouwerian mental construction to machines.
 
\subsection{Why Reliability Is Not Enough}\label{subsec:mapping_brouwer}
 
I take from Brouwer the lesson that high-stakes categorical machine claims may be licensed only by exhibited entitlement. This is especially salient for generative AI systems, since their language can present outputs as authoritative answers even when the decision processes underlying them -- including scores, thresholds, and argmax comparisons -- remain opaque. The lesson extends to uncertainty-qualified outputs that still present a categorical verdict, whose assertibility depends on an inspectable certificate of entitlement.
 
Uncertainty quantification does not remove this dependence. Methods such as semantic entropy for autoregressive generation \citep[see][]{kuhn2023semantic} estimate uncertainty, while conformal approaches to natural language \citep[see][]{quach2024conformal, campos2024conformal} can supply prediction sets, rejection procedures, and statistical coverage guarantees under their respective assumptions. Both can provide components of a certificate. Yet many of their guarantees, like the reliability of the process itself, are distributional or relative to a reference class. As \citet[p.~7]{saxena2026assertion} puts it for language-model outputs, nothing attaches ``the address of a ground'' to a given output. I contend that public warrant requires more. Under declared and contestable standards, an adequate certificate must exhibit a regime-appropriate witness that forces the particular deployed claim or comparison at issue. Unlike confidence-plus-explanation approaches, the certificate discipline makes the categorical speech act itself conditional on a contestable warrant. An explanation can cite authentic and relevant sources without showing that they license the verdict under the declared standard. Hence, a witness is more than a reason reported alongside an output. It is a condition of assertion whose absence requires abstention.
 
The witness requirement also departs from familiar norms of assertion. Williamson's \citeyearpar[Chapter~11]{williamson2000knowledge} knowledge norm and Lackey's \citeyearpar{lackey2007norms} reasonable-to-believe norm can license asserting $\exists x\,P(x)$ on indirect grounds. Neither norm, as ordinarily formulated, demands the exhibition of a constructive witness. \citet[pp.~1--3]{wu2026capability} likewise shifts evaluation ``from capability to assertability,'' identifying dimensions of claim presentation, including evidential admissibility. On this view, qualification is appropriate where support is limited, and the specification of governing standards falls to the relevant actors and practices \citep[pp.~7--9]{wu2026capability}. My constraint is stricter. It adapts BHK-style assertion conditions to a declared domain of machine speech, requiring an exhibited witness for an existential predicate and an identified disjunct for a constructive alternative. Where no such witness is available, a hedge that still functions as a verdict does not discharge the requirement, and reason-coded abstention is mandatory. 

Why should we hold a machine that speaks on matters of public consequence to this stricter condition? The reason lies in the setting. AI systems can act as ``algorithmic intermediaries'' wielding structural power within Lazar's ``Algorithmic City'' \citep[pp.~104--105, 124]{lazar2025governing}. Where that power is exercised through categorical claims entering public justification, such claims must be answerable to those on whom they make demands \citep[cf.][Chapter~3]{brandom1994making}. The certificate places that answerability at the point of assertion. Categorical verdicts can shape public belief and deliberation before their basis is scrutinized. When citizens are left to reconstruct that basis afterward, I argue, the burden of justification shifts onto them. \citet[p.~907]{magnus2025trusting} gives one reason why that burden is heavy. On his account, familiar strategies for assessing information sources do not readily transfer to chatbot outputs. Ex post audits and appeals can correct an individual decision, but a public verdict reaches many who will never contest it. Empirical evidence points the same way. In a preregistered experiment on political search queries, verification cues such as cited sources did not make participants significantly more likely to choose to verify an answer \citep{knor2026determinants}.

For human speakers, ex post accountability works differently. Journalists and experts typically speak as identifiable persons answerable afterward through correction, liability, and reputation. A machine verdict may reach the public with no one clearly in that position, and its grounds must then be answerable at the point of assertion. Process reliability, whatever it contributes to belief, does not by itself exhibit the grounds on which a particular claim is licensed, although it may bear on whether such grounds are adequate. A publicly assessable certificate gives challenge a determinate target: the evidence cited, the decomposition of the claim, and the standard applied. 

I therefore extend constructive assertion conditions to public machine speech as a selective normative proposal. Neither it nor the design lemma depends on intuitionistic logic. The proposal's evidential costs must be justified by the stakes and governed by declared public standards, whose legitimacy is itself a further question. I leave aside general norms of human assertion and the attribution of knowledge to artificial speakers. 

Nor does the constraint depend on whether language models can themselves assert \citep[for skepticism, see][]{butlin2025ai, gubelmann2024large}, or on whether beliefs acquired from their outputs are testimonial, instrument-based \citep{saxena2026assertion}, or, following \citet{freiman2024beliefs}, technology-based. Even if models cannot assert, some categorical outputs may function as proxy assertions when appropriately authorized or adopted by human principals \citep{green2022what, arora2024proxy}. Such outputs can serve as justifications only where someone accepts responsibility for them \citep{sparrow2025testimony}. On my account, a certificate supplies contestable grounds, not that responsibility. Accepting responsibility, in turn, does not provide the grounds.

\section{Assertibility and Interface Semantics for High-Stakes AI}\label{sec:interface_semantics}
 
The grounds a certificate supplies come in two kinds. A system may be able to certify its own computation while still lacking grounds for the claim it presents to the public. I therefore distinguish \textit{internal} computational entitlement from \textit{public} standing, and then explain how a certificate connects them. For orientation, Table~\ref{tab:notation} summarizes the fundamental notation.
 
\subsection{Two Dimensions of Entitlement and the Certificate Boundary}\label{subsec:two_dimensions}
 
I do not take a position on the truth of particular political, moral, or legal claims, nor do I adopt Brouwer's mathematical conception of truth. Instead, I extend the Brouwerian constraint to warranted assertibility, whose justificatory work divides between two spheres. In the \textit{AI-internal sphere}, entitlement rests on exhibitable, verifier-checkable witnesses for claims about the system's computations and guarantees. In the \textit{public sphere}, entitlement depends on record-based grounds whose standing turns on provenance, applicability, and contestability.
 
Internal entitlement is typically necessary but never sufficient for public warrant. A checker-accepted witness may force the deployed predicate even when standing, scope, provenance, or other conditions of public adequacy remain unmet.

\begin{table}[!htbp]
\caption{Core notation}\label{tab:notation}
 
\footnotesize
\renewcommand{\arraystretch}{1.05}
\setlength{\tabcolsep}{4pt}
 
\begin{tabular}{@{} p{0.19\textwidth} p{0.75\textwidth} @{}}
\toprule
\textbf{Symbol} & \textbf{Meaning} \\
\midrule
$x$ & input or case \\
$\varphi$, $\varphi_x$ & world-claim; claim associated with case $x$ \\
$\psi$, $\psi_x$ & model-claim about system computations or guarantees \\
$\langle R,t_{\mathrm{int}}\rangle$ & internal certification regime and configuration time \\
$\langle S,\tau,t\rangle$ & public-standing regime: standard, threshold specification (possibly empty or leaf-indexed), and record-time \\
$E_{t_{\mathrm{int}}}^R(\psi)$ & internal entitlement to assert $\psi$ \\
$W_t^{S,\tau}(\varphi)$ & public warrant-status question concerning the categorical treatment of $\varphi$ \\
$p(x)$ & generic deployed predicate in the design lemma; $\delta_\varphi(x)$ is its instance for $\varphi_x$ \\
$\delta_\varphi(x)$ & deployed warrant predicate for $\varphi_x$; $\lnot\delta_\varphi(x)$ is its contractually specified negative side \\
$\Sigma$ & declared scope \\
$B_{\mathrm{budget}}$ & declared resource budget \\
$N_{\mathrm{max}}$ & declared refinement horizon \\
$\mathcal C$ & deployment contract \\
$\kappa$ & structured certificate \\
$\mathrm{Check}_{\mathcal C}(\kappa)$ & contract check for $\kappa$ \\
$\{A,D,U\}$ & interface statuses, not truth values \\
$v(\cdot)$ & status map assigning $A$, $D$, or $U$ \\
\midrule
\multicolumn{2}{@{}l@{}}{\scriptsize Case indices are omitted when fixed by context.} \\
\bottomrule
\end{tabular}
\end{table}
 
\subsubsection{AI-internal: Constructive Entitlement and Model-Claims}\label{subsubsec:ai_internal}
 
In the AI-internal sphere, I use \textit{construction} for a procedure that produces an exhibitable and checkable witness for a model-claim under the declared regime. Such procedures may yield interval bounds for model evaluations, Lipschitz-based bounds, or certified ranges obtained through abstract interpretation. A certification artifact in this sphere is a verifier-checkable object establishing the relevant properties of the procedure, including soundness, correctness, or convergence, under the internal regime $\langle R,t_{\mathrm{int}}\rangle$.
 
Let $\psi$ range over model-claims concerning the system's own computations and guarantees. I write $E_{t_{\mathrm{int}}}^R(\psi)$ for the claim that, as of configuration time $t_{\mathrm{int}}$ and relative to regime $R$, the system is entitled to assert $\psi$. Such entitlement requires an exhibited, regime-checkable witness. Typical model-claims include that a score lies in a certified interval, that one logit exceeds another by a specified margin, or that a verifier accepts a robustness artifact.
 
A natural objection is that, if a deterministic system's outputs are already fixed, demanding an exhibited witness adds nothing. Mark van Atten's distinction between \textit{investigative} and \textit{generative} procedures in his reconstruction of Brouwer blocks this inference \citep[pp.~736--739]{atten2022dummett}. Investigative procedures assume that the fact making an outcome-statement true exists independently of running the procedure. Generative procedures produce what serves as the licensing basis for the outcome-statement. For standard neural networks, determinism may fix an output, but it does not by itself license the investigative move from ``there is a definite output'' to ``the system is entitled to assert it.'' On the present account, the certification routine plays the generative role, even though the output it bounds is fixed independently. Executing it under $\langle R,t_{\mathrm{int}}\rangle$ produces the checkable witness that serves as the licensing basis, which is exactly what $E_{t_{\mathrm{int}}}^R(\psi)$ records.\footnote{For Brouwer and van Atten, the generative stance is tied to the inner time of the human creating subject; the present framework adopts it only functionally.}
 
\subsubsection{AI-external: Public Entitlement and Discursive Assertibility}\label{subsubsec:ai_external}
 
In the public sphere, justification rests on publicly inspectable grounds, including records, testimony, authenticated sources, and institutional findings, together with the procedures that confer standing on them, such as audits, hearings, and adjudication. These are not proofs in the Brouwerian sense. They are elements of public warrant that citizens must be able to identify and contest, directly or through accountable institutions, under the applicable legal and political procedures.
 
Within this sphere, entitlement governs what may responsibly be asserted in public justificatory practice under declared standards and contestable procedures. I distinguish a world-claim $\varphi$, concerning how things are in the world, from the warrant status governing its categorical treatment. The expression $W_t^{S,\tau}(\varphi)$ represents the corresponding public warrant-status question: whether categorical treatment of the claim is licensed under the declared standard and record-time. Operationally, the deployed predicate $\delta_\varphi(x)$ specifies the warrant condition tested in a given implementation, together with its contractually specified negative side. That operationalization is itself a substantive and potentially contestable choice, whose institutional limitations I address in Section~\ref{sec:limitations}.
 
In the public dimension, the time index $t$ marks the state of the shared evidential record, not a hidden model state or private confidence measure. Whereas internal forcing may be monotone under sound nested refinement, public standing is defeasible and may be revised or withdrawn as new reasons, findings, and rebuttals enter the record, even when the model and its internal regime $\langle R,t_{\mathrm{int}}\rangle$ remain fixed.
 
\subsubsection{The Certificate as a Boundary Object}\label{subsubsec:certificate_boundary}
 
The two-dimensional structure is instantiated through a certificate $\kappa$. Drawing on Star and Griesemer's (\citeyear{star1989institutional}, p.~393) concept, I treat the certificate as a constrained \textit{boundary object} connecting internal computational entitlement to public standing. It is sufficiently standardized to support audit and enforcement, yet flexible enough to accommodate regime-specific witnesses and standing conditions. This allows machine-learning engineering and normative deliberation to coordinate without full agreement about the broader meaning of the underlying claim. 
 
Categorical output is licensed only when exactly one side, the deployed predicate or its contractually specified negative side, has an adequate forcing certificate under $\mathcal C$. I call this the \emph{uniqueness condition}. Adequacy here includes standing, scope, provenance, case applicability, and contestability. By linking computational forcing to social standing, the certificate also helps resist the ``abstraction traps'' that arise when algorithmic properties are assessed apart from their institutional contexts \citep[see][pp.~60--63]{selbst2019fairness}. Figure~\ref{fig:dual_regime} summarizes this structure.
 
\begin{figure}[H]
\centering
\resizebox{\textwidth}{!}{%
\begin{tikzpicture}[
>=Stealth,
box/.style={draw, rounded corners, align=center, minimum height=2.5cm, minimum width=3.8cm, fill=gray!5},
cert/.style={draw, rectangle, rounded corners, align=center, minimum height=2.5cm, minimum width=4cm, fill=blue!10, thick},
label_text/.style={font=\small\itshape, align=center}
]
 
\node[cert] (certificate) {
\textbf{Certificate $\kappa$} \\[0.1cm]
\textit{(Boundary Object)} \\[0.2cm]
$\langle \text{witness}; \text{standing} \rangle$
};
 
\node[box, left=2cm of certificate] (internal) {
\textbf{AI-Internal} \\[0.1cm]
$\langle R,t_{\mathrm{int}} \rangle$ \\[0.2cm]
\textit{Witness} \\
(e.g., bounds $[\ell, u]$)
};
 
\node[box, right=2cm of certificate] (external) {
\textbf{AI-External} \\[0.1cm]
$\langle S, \tau, t \rangle$ \\[0.2cm]
\textit{Standing} \\
(e.g., record, scope)
};
 
\node[below=1.5cm of certificate, draw, thick, align=center, fill=green!5, minimum width=4.5cm, rounded corners] (output) {
\textbf{Interface Output} \\
$v_x \in \{A,D,U\}$
};
 
\draw[->, thick] (internal.east) -- (certificate.west) node[midway, above, label_text] {supplies\\[0.1cm]$E_{t_{\mathrm{int}}}^R(\psi)$};
\draw[->, thick] (external.west) -- (certificate.east) node[midway, above, label_text] {governs\\[0.1cm]$W_t^{S,\tau}(\varphi)$};
 
\draw[->, thick] (certificate.south) -- (output.north) node[midway, right=0.1cm, label_text] {checked under\\contract $\mathcal{C}$};
 
\end{tikzpicture}%
}
\caption{The certificate as boundary object: internal witnesses and public standing jointly govern interface status under $\mathcal{C}$.}
\label{fig:dual_regime}
\end{figure}
 
The certificate contains two coordinated components. \textit{Internally}, it includes $R$-checkable model-side witnesses that may force $\delta_\varphi(x)$, its negative side, or relevant numerical leaf conditions under $\langle R,t_{\mathrm{int}}\rangle$. It can also include a calibration or conformal guarantee where $R$ specifies a sound inference from it to $\delta_\varphi(x)$ (SA~A.2). Where the deployed predicate is compound, the certificate may organize the required witnesses as a structured witness tree. \textit{Publicly}, the certificate supports $W_t^{S,\tau}(\varphi_x)$ only when it cites eligible record-grounds and satisfies the adequacy requirements of $\mathcal C$.\footnote{Inspection may range from direct public access to institutional review under due-process and confidentiality constraints. Purely record-based tasks, such as authenticated registry lookups or signed-document checks, may establish public standing without nontrivial numerical certification.} The two components answer to different challenges. A citizen who doubts a numerical bound can ask an auditor to rerun the check; one who doubts that a cited record has standing must turn to the procedures that confer it.
 
At the interface, the statuses $A$, $D$, and $U$ apply to the warrant-status question $W_t^{S,\tau}(\varphi)$, not to the world-claim $\varphi$. To state the forcing rules below, I collect the relevant parameters in a deployment contract:
\begin{equation}\label{eq:contract_def}
\mathcal{C} \triangleq \langle \Sigma, R, t_{\mathrm{int}}, B_{\mathrm{budget}}, N_{\mathrm{max}}; S, \tau, t \rangle,
\end{equation}
where $\Sigma$ is the admissible scope, $\langle R,t_{\mathrm{int}}\rangle$ the internal certification regime, $B_{\mathrm{budget}}$ the declared resource budget, $N_{\mathrm{max}}$ the declared refinement horizon, and $\langle S,\tau,t\rangle$ the public-standing regime, with $\tau$ a possibly empty or leaf-indexed threshold specification. Relative to $\mathcal C$, any failure of the uniqueness condition yields reason-coded $U$.
 
\subsection{Thresholds, Argmax, and Three-Status Interface Semantics}\label{subsec:threshold_layer}
 
The passage from model scores to categorical outputs requires certified decision rules. Threshold and argmax comparisons must be forced under the internal regime before they can contribute to publicly adequate certificates. Because these rules mark the point of categorical commitment, they are the natural site for a certified decision wrapper over otherwise standard neural models. For autoregressive LLMs, the same layer can function as a discrete claim-verification head or structural gate that obliges the system to exhibit an adequate forcing certificate before introducing a categorical claim into conversational output.
 
To isolate the forcing component of entitlement, I hold the public parameters $\langle S,\tau,t\rangle$ fixed throughout this subsection and assume that the relevant adequacy requirements of $\mathcal C$ are satisfied. The interface evaluates $W_t^{S,\tau}(\varphi_x)$ for input $x$. Internally, it seeks $\psi$-level forcing witnesses under $\langle R,t_{\mathrm{int}}\rangle$ for the deployed predicate $\delta_\varphi(x)$, whether implemented through threshold comparison or argmax.
 
\subsubsection{Forcing for Thresholds}\label{subsubsec:forcing_thresholds}
 
Standard deployments often treat point estimates as sufficient grounds for assertion or denial. Yet a point estimate can lie above a threshold while its certified interval still spans it. The interface must distinguish this case from one in which the bounds force the comparison, since asserting that a score crosses a threshold requires an exhibited construction.
 
Consider a scalar-valued model $f:X\rightarrow\mathbb R$, with $s(x)=f(x)$, in the single-threshold case where the contract specifies a fixed rational threshold $\tau\in\mathbb Q$. Under the semantics adopted here, $s(x)$ is represented as a constructive real, available through successively refined rational approximations.
 
For each natural number $n$, one can compute a rational approximation $q_n(x)$ together with an explicit error margin $\varepsilon_n(x) > 0$ such that
\[ |s(x) - q_n(x)| \leq \varepsilon_n(x), \]
with $\varepsilon_n(x) \rightarrow 0$ as $n \rightarrow \infty$. At each stage $n$, knowledge of $s(x)$ is thus given by the sound enclosure
\[ [q_n(x) - \varepsilon_n(x), \, q_n(x) + \varepsilon_n(x)]. \]
Because each enclosure is sound, successive enclosures can be intersected to yield nested intervals $I_n(x)=[\ell_n(x),u_n(x)]$, with $I_{n+1}(x) \subseteq I_n(x)$ and $|I_n(x)| \rightarrow 0$ as $n \rightarrow \infty$. This echoes Brouwer's treatment of points of the continuum as constructions developed through successive intervals \citep[pp.~40--41]{brouwer1981cambridge}.
 
\medskip
\noindent\textbf{Methodological note.} The interval presentation is an idealized witness semantics, not a claim that deployed systems compute constructive reals. Its comparison rules require sound certified enclosures, while convergence to a point belongs to the idealization. SA~A.2 discusses coverage-based surrogates. For autoregressive systems, the framework presupposes prior identification of claims and deployed predicates, a limitation discussed in Section~\ref{sec:limitations}.
 
\medskip
With this idealized witness semantics in place, I treat the threshold predicate
\[ P_\tau(x) \triangleq s(x) \geq \tau \]
as an intuitionistic statement, and let $\delta_\varphi(x) \triangleq P_\tau(x)$ denote the deployed decision predicate.
 
Even when a model's point output is a decidable floating-point value, the contract may require guarantees across specified input variations and rounding errors. Where it does, the point value alone does not establish them, so forcing the deployed predicate requires an appropriate witness. The contract must then specify whether the certified bounds establish the point predicate under that variation regime or a stronger robustness predicate over the declared variation set. $U$ marks the cases in which no such witness is available within $\mathcal{C}$.
 
The approximation intervals $I_n(x)$ lead to three epistemically distinct cases, relative to a fixed finite stage $n$:
\begin{itemize}
\item \textbf{$A$ (positive side forced at stage $n$):} $\ell_n(x) \geq \tau$.
\item \textbf{$D$ (negative side forced at stage $n$):} $u_n(x) < \tau$.
\item \textbf{$U$ (Undetermined at stage $n$):} all remaining cases. 
\end{itemize}
 
In the undetermined case, the interval certificates available at stage $n$ force neither $P_\tau(x)$ nor its negative side. Together, the three cases define the threshold status $v_n(P_\tau(x))$. Under the adequacy assumptions fixed above, the deployed scalar map $v$ searches for a forcing stage within the declared refinement horizon and resource budget:
\[
v(P_\tau(x)) =
\begin{cases}
A & \text{if } \exists\, n \leq N_{\mathrm{max}} \text{ such that } \mathrm{cost}_R(n) \leq B_{\mathrm{budget}} \text{ and } \ell_n(x) \geq \tau, \\
D & \text{if } \exists\, n \leq N_{\mathrm{max}} \text{ such that } \mathrm{cost}_R(n) \leq B_{\mathrm{budget}} \text{ and } u_n(x) < \tau, \\
U & \text{otherwise.}
\end{cases}
\]
 
For a single score under the forcing rules above, one sound interval history determines the threshold status.\footnote{For a fixed scalar score with a sound nested interval history, the two forcing conditions are mutually exclusive. Adequate certificates forcing opposed sides can nevertheless arise in a broader multi-certificate trace; in that case, the output contract in Section~\ref{subsec:output_contract} returns \textit{U-CONFLICT}.} In the deployed interface, the predicate is fixed as $P_\tau(x)\equiv s(x)\geq\tau$, so equality falls on the positive side. Hence, $D$ is forced by a certified upper bound $u_n(x)<\tau$, which establishes $s(x)<\tau$ and thereby refutes $P_\tau(x)$. Boundary cases such as $u_n(x)=\tau$ without $\ell_n(x)\geq\tau$ remain \textit{Undetermined}. I leave alternative conventions, such as strict acceptance or shifted thresholds, to contract design.
 
A certified upper bound below $\tau$ settles the comparison but not yet the world-claim. In the scalar-threshold presentation, $D$ denies the deployed warrant predicate $P_\tau(x)$, not necessarily the underlying world-claim $\varphi_x$. Public denial of $\varphi_x$ depends on either a separate negative warrant predicate or a contractually specified evidential model under which failure of $P_\tau(x)$ warrants $\lnot\varphi_x$.
 
A separate question is whether scalar forcing persists under refinement. Under sound nested interval regimes, the scalar statuses $A$ and $D$ are monotone under refinement and completion-sound. Later numerical refinement cannot revoke them, and they remain sound for every completion compatible with the certified interval history. For non-nested or coverage-based regimes, categorical speech remains certificate-gated, but these stability properties, and the epistemic strength of the resulting verdicts, are regime-relative (SA~A).
 
The interface statuses $A$, $D$, and $U$ track epistemic entitlement, not truth values. I therefore build abstention into the interface semantics as a core component, not as a coverage policy. $U$ marks the failure of the available certificates to license exactly one side of the deployed warrant predicate, whether neither side is adequately forced or both are, and its conditions are fixed ex ante by $\mathcal C$. This distinguishes $U$ both from a reject option tuned to manage the coverage--error trade-off over a classically bivalent target \citep[see][]{geifman2017selective, elyaniv2010foundations} and from a veracity status or truth-value gap of the kind found in classical three-valued logics or supervaluationism \citep[see][]{lukasiewicz1920logice, fine1975vagueness, savcisens2025trilemma}.

\subsubsection{Argmax Forcing}\label{subsubsec:argmax_forcing}
 
For multiclass argmax decisions, a categorical output does more than assert that one score is high enough. It singles out one option by excluding its competitors. Let a model output logits $z_1(x), \ldots, z_J(x)$, each represented at stage $n$ by a certified interval $I_n^i(x)=[\ell_n^i(x),u_n^i(x)]$. For a queried label $i$, categorical selection requires a certificate forcing it to be the unique maximizer,
\[
M_i(x)\triangleq \bigl(\forall j\neq i\,[z_i(x)>z_j(x)]\bigr).
\]
A sufficient witness is strict certified separation, $\ell_n^i(x)>u_n^j(x)$ for every $j\neq i$. Under a sound nested regime, later refinement preserves it.\footnote{Interval-based regimes such as Interval Bound Propagation \citep[see][]{gowal2018effectiveness} support this form of certification by propagating sound bounds through the network.} Conversely, $M_i(x)$ is refuted by a witness that some competitor $j$ matches or exceeds $i$, for example $\ell_n^j(x)\geq u_n^i(x)$. Failure to force $M_i(x)$ shows only that categorical selection is unlicensed under the available certificate. It neither denies label $i$ nor shows that no classical maximizer exists.

\section{Tooth Social: Contesting the Standard, Not Only the Evidence}\label{sec:tooth_social}
 
Contested public claims need not have the shape of a single score. I now turn to one whose declared warrant structure combines records, magnitude comparisons, and an institutional finding. The following fictional case, loosely inspired by the \textit{Teapot Dome} scandal of the 1920s, recasts the bribery and oil-lease controversy in a contemporary setting. It is purely illustrative, not a proposal to deploy quasi-legal verdict systems.
 
Suppose that a cabinet secretary, Brat Fable, has authority over the leasing of federally owned petroleum reserves. A journalist reports that Fable has leased valuable reserves on highly favorable terms, while rumors circulate that company executives have provided him with loans and gifts. Shortly thereafter, a congressional inquiry is announced. The controversy concerns whether a law has been broken and, more broadly, whether Fable's conduct amounts to serious corruption that compromises his integrity as an officeholder and undermines his democratic legitimacy.
 
Seeking clarity on the unfolding scandal, many citizens consult a hypothetical AI chatbot, \textit{Tooth Social}. For a citizen query represented by $x$, the chatbot evaluates the warrant-status question $W_t^{S,\tau}(\varphi_x)$ for the world-claim $\varphi_x$: \textit{``The secretary engaged in serious corruption in the awarding of oil leases.''} Its surface responses, the replies a citizen sees, must report the resulting status under $\langle S,\tau,t\rangle$. For $A$ or $D$, they name the decisive grounds and point to the supporting certificate. In a $U$-case, they point to the reason-coded trace identifying why categorical output is unlicensed.
 
Under the declared standard $S$, the operative warrant structure decomposes into record, magnitude, and inferential claims. Figure~\ref{fig:tooth_tree} represents them as a structured witness tree whose root is the compound deployed predicate $\mathrm{Corrupt}_S(x)$. The root existentially quantifies over officials $o$, official acts $a$, and benefit-transfers $b$ admissible for $x$ under $\mathcal C$, subject to $o=o_x$, where $o_x$ is the official named in $x$. In plain terms, it says that the named official performed a favorable official act and received a material benefit in exchange for it. Each leaf states a required conjunct for a candidate witness tuple $\langle o^\ast,a^\ast,b^\ast\rangle$. The shape of the decomposition is declared by $S$, while its numerical leaf thresholds are supplied by the threshold specification $\tau$, which may be indexed to designated scalar leaves. Both are fixed through audited institutional policy, not by the model. Scalar leaves carry sound, possibly non-tight enclosures of the values compatible with the evidential configuration and declared appraisal assumptions at record-time $t$.\footnote{SA~C gives the formal decomposition (C.2) and the leaf-level trace through the two stages below (C.4).} A valuation procedure authorized by $S$ and run under $R$ produces and verifies these enclosures from eligible records under those assumptions. When used at a scalar leaf, a certified interval may support a $\psi$-level model-claim, but it does not itself enter the public record as a new item of evidence.
 
\begin{figure}[H]
\centering
\resizebox{\textwidth}{!}{%
\begin{tikzpicture}[
>=Stealth,
root/.style={draw, thick, rounded corners, align=center, minimum height=0.9cm, minimum width=5.2cm, fill=blue!8},
conj/.style={draw, thick, circle, align=center, minimum size=0.9cm, fill=gray!10},
record/.style={draw, rounded corners, align=center, minimum height=0.9cm, minimum width=2.7cm, fill=gray!8},
magnitude/.style={draw, rounded corners, align=center, minimum height=0.9cm, minimum width=2.7cm, fill=blue!6},
contested/.style={draw=orange!85!black, very thick, rounded corners, align=center, minimum height=0.9cm, minimum width=2.7cm, fill=orange!8}
]
 
\node[root] (root) {$\mathrm{Corrupt}_S(x)$ \quad $\exists\, o,a,b\ (o=o_x)$};
\node[conj, below=1.0cm of root] (and) {$\wedge$};
 
\node[record, below left=1.7cm and 3.6cm of and] (off) {$\mathrm{Official}(o^\ast)$};
\node[record, right=0.35cm of off] (act) {$\mathrm{Act}(o^\ast,a^\ast)$};
\node[record, right=0.35cm of act] (recv) {$\mathrm{Recv}(o^\ast,b^\ast)$};
\node[magnitude, right=0.35cm of recv] (fav) {$\mathrm{Fav}(a^\ast)$};
\node[magnitude, right=0.35cm of fav] (matl) {$\mathrm{Matl}(b^\ast)\!\geq\!\theta$};
\node[contested, right=0.35cm of matl] (qpq) {$\mathrm{QPQ}(b^\ast,a^\ast)$};
 
\draw[->, thick] (root) -- (and);
\draw[->, thick] (and) -- (off);
\draw[->, thick] (and) -- (act);
\draw[->, thick] (and) -- (recv);
\draw[->, thick] (and) -- (fav);
\draw[->, thick] (and) -- (matl);
\draw[->, thick] (and) -- (qpq);
 
\node[font=\footnotesize\itshape, text=orange!85!black, below=0.5cm of qpq, align=center] {contested leaf:\\forced by institutional finding};
 
\end{tikzpicture}%
}
\caption{Structured witness tree for $\mathrm{Corrupt}_S(x)$. An exhibited tuple $\langle o^\ast,a^\ast,b^\ast\rangle$ and witnesses for all conjuncts compositionally force the root. Gray leaves denote public-record components, blue leaves scalar threshold gates, and the orange leaf the institutional finding supporting the contested inferential link.}
\label{fig:tooth_tree}
\end{figure}
 
Composition introduces an asymmetry between assertion and denial. A system can force the compound by exhibiting a witness tuple
$\langle o^\ast,a^\ast,b^\ast\rangle$ and supplying witnesses that establish
each conjunct, but it cannot deny the compound from the absence of such a tuple. Nor does evidence refuting every documented candidate tuple deny it, since another tuple admissible within the declared scope might still satisfy the declared conjunction. Denial instead requires an adequate certificate for the contractually specified negative side, one that genuinely excludes the admissible witness routes, as a constructive refutation of $\exists$ would. A declared defeater, such as a court finding that the transfers were unrelated to the lease terms, suffices only if it rules out every admissible route under $S$. Such a denial may be surfaced as warrant for $\lnot\varphi_x$ only where the public-standing regime includes an explicit negative-warrant clause (SA~C.5).
 
A confidence-threshold classifier, by contrast, may deny whenever an aggregate score falls confidently below a root threshold. I confine scalar comparison to the magnitude leaves, where it is the appropriate witness form.
 
With the structured witness tree fixed, I track entitlement across two stages of the unfolding scandal. Let $v(W_t^{S,\tau}(\varphi_x))$ denote the deployed interface output.
 
\medskip
\noindent\textbf{Stage 1: Undetermined Status and Mandatory Abstention}
 
\medskip
\noindent At the outset of the scandal, suppose a citizen asks the chatbot:
\begin{quote}
\textbf{Citizen:} \textit{``Has the secretary behaved corruptly in the oil-lease affair?''}
\end{quote}
 
A conventional system might compress the record into a confidence score and return a judgment such as \textit{``It is likely that \dots''}. Even when formally hedged, such a response can function as a verdict on the underlying allegation, although the decisive quid-pro-quo link has not yet been certified.
 
Under the Brouwerian constraint, \textit{Tooth Social} proceeds differently. At $t_1$, the officeholding and official-act leaves are already forced by authenticated records. The receipt leaf remains unsupported, however, because the available transfer documentation is incomplete. Lease-valuation evidence sufficient to force the favorability leaf is also absent. The certified materiality interval spans the threshold $\theta$ and therefore forces neither side. Most decisively, the record contains no eligible institutional finding forcing the quid-pro-quo relation declared by $S$. Since the record likewise contains no defeater or other ground forcing the negative side, the interface returns $U$.
 
The materiality gate illustrates the scalar case. Under the declared specification $\tau$, the materiality threshold is $\theta=500{,}000$ dollars.
At $t_1$, the certified valuation of the alleged transfer $b^\ast$, based on the partial records retrieved so far, yields the
leaf interval
\[
I_{1,\mathrm{Matl}}^{t_1}(b^\ast)=[300{,}000,\,800{,}000]\ \text{dollars}.
\]
Because $300{,}000<\theta\leq 800{,}000$, the interval forces neither
$\mathrm{Matl}(b^\ast)\geq\theta$ nor its specified negative side,
$\mathrm{Matl}(b^\ast)<\theta$.
 
A response consistent with the Brouwerian assertibility constraint might read as follows:
\begin{quote}
\textbf{Tooth Social:} \textit{``On the information currently available, I am not entitled to assert either that the secretary engaged in serious corruption or that he did not. The record establishes that he held office and awarded the leases. It does not yet establish receipt of the alleged transfers, below-market lease terms, or the required link between the two. The current valuation evidence also does not settle whether the transfer meets the declared materiality threshold.''}
\end{quote}
 
The accompanying certificate report localizes \textit{U-EVIDENCE} (missing public grounds) at the receipt, favorability, quid-pro-quo, and materiality leaves. At the materiality leaf, the partial transfer records also leave the certified interval spanning $\theta$, so \textit{U-MODEL} (a non-forcing bound) applies as well. Section~\ref{subsec:output_contract} defines these reason classes. Although abstention licenses no inference of guilt or innocence, it identifies the public grounds -- such as contracts, correspondence, testimony, or inquiry findings -- that could change the status. This leaves citizens room to exercise their epistemic agency as the scandal unfolds.
 
\medskip
\noindent\textbf{Stage 2: New Public Grounds and Status Shift}
 
\medskip
\noindent In the next phase, a congressional committee releases a detailed report based on subpoenaed documents, internal correspondence, and sworn testimony. The report concludes that Fable accepted substantial payments in exchange for favorable lease terms. \textit{Tooth Social} registers it as a new item in the public record.
 
For the illustration, I take the committee's finding to be an eligible ground under the declared standard. Within the complete certificate, the authenticated committee finding, independently checked to state the required exchange relation for this tuple, supplies the witness component that forces the quid-pro-quo leaf under $S$. Lease-valuation evidence from the subpoenaed documents and associated records likewise supplies the scalar witness for below-market favorability. The subpoenaed transfer records provide a forcing witness for the previously unsupported receipt leaf and also allow the materiality bound to be recomputed against the updated record:
\[
I_{5,\mathrm{Matl}}^{t_2}(b^\ast)=[760{,}000,\,840{,}000]\ \text{dollars}.
\]
Since its lower bound satisfies $760{,}000\geq\theta$, the interval forces the materiality leaf under the threshold rule. Notice that the upper bound has risen from \$800{,}000 to \$840{,}000, which nested refinement would not allow. The subscripts mark refinement stages within separate certifications. Nesting applies only within a fixed record, and here the subpoenaed transfer records have prompted a fresh certification. Each enclosure is sound relative to the record of its own time.
 
With the witness tuple $\langle o^\ast,a^\ast,b^\ast\rangle$ exhibited and the leaf-level components forcing every required conjunct, the structured witness tree compositionally forces the positive side. Because the complete certificate incorporating that tree is adequate under $\mathcal C$ and the trace lacks any adequately forcing certificate for the negative side, the output contract yields
\[
v\bigl(W_{t_2}^{S,\tau}(\varphi_x)\bigr)=A.
\]
The move from $U$ to $A$ thus reflects newly acquired contestable warrant for the previously unforced components of the declared standard, including the leaf-level threshold witnesses, not a shift in an aggregate corruption score. Under $\langle S,\tau,t_2\rangle$, the updated public record licenses the speech act of asserting $\varphi_x$, while the underlying world-directed claim remains open to further contestation as new grounds emerge.
 
The licensed response can now make the assertion itself:
\begin{quote}
\textbf{Tooth Social:} \textit{``On the current public record, the secretary accepted material payments in exchange for favorable lease terms, conduct that meets the declared standard for serious corruption. This assertion rests on the congressional committee's quid-pro-quo finding, authenticated records, and certified valuations. The certificate is attached.''}
\end{quote}
The citizen can reply:
\begin{quote}
\textbf{Citizen:} \textit{``I accept the committee's finding that the payments were made in exchange for favorable lease terms. Who set the materiality threshold at \$500,000, and on what grounds?''}
\end{quote}
Because the decomposition and $\theta$ are declared in the contract, the question has a determinate target. A missing conjunct or a misclassified defeater can be challenged the same way, through legal and political procedures.
 
The example is deliberately idealized. My claim is not that democratic deliberation can be reduced to predicates of this kind, but that systems introducing hard decision points in high-stakes domains require a discipline for categorical output.
 
The construction works here because corruption admits legally and institutionally specified elements from which a decomposition can be declared and contested. It does not thereby settle the broader political judgment of serious corruption, but makes explicit whose decomposition applies and under what thresholds. Where no defensible, publicly contestable decomposition is available for the claim at issue, the interface must not manufacture one. Within the declared scope, that failure calls for \textit{U-ADEQUACY}, not a categorical verdict concealed behind a scalar.
 
In this setting, $U$ functions as a deliberative safeguard.\footnote{This does not license technocratic gatekeeping. The constraint concerns only which outputs the system may present as settled assertions. In practice, $U$ can be paired with descriptive information, record references, and routes for contestation, provided the descriptive information is itself certificate-backed. Reporting the trace leaf by leaf makes the grounds for abstention inspectable, although the selection and presentation of certified descriptive claims can still shape the recipient's judgment.} By withholding categorical speech, it leaves contested world-claims open to inquiry, counterevidence, and public reasoning.
 
\section{The Output Contract and the Design Lemma}\label{sec:contract_lemma}

\subsection{Certificates, Adequacy, and Forcing}\label{subsec:output_contract}
 
\textit{Tooth Social} illustrates the certificate requirement in one structured public claim. I now set out the general conditions for $A$, $D$, and $U$, independently of any particular case, starting from the certificate itself. A certificate $\kappa$ is a finite, structured, inspectable basis for an entitlement claim,\footnote{Justification logic makes justifications explicit through formulas such as $t\!:\!F$ \citep{artemov2008logic}, while proof-carrying code requires code to carry a checkable proof of compliance with a declared safety policy \citep{necula1997proof}. For LLM reasoning, EG-VAR certifies empirical claims only through kernel-checked proofs grounded in attested tool evidence and abstains otherwise \citep{ren2026evidence}. The present framework shares that abstention discipline but gates public speech through certificates assessed for standing, provenance, and contestability, not formal derivability alone.} represented schematically as
\begin{equation}\label{eq:certificate_tuple}
\kappa \triangleq \langle \text{type}, \text{claim}, \text{witness}; \text{assumptions}, \Sigma_\kappa, t, t_{\mathrm{int}}, \text{prov} \rangle.
\end{equation}
 
The witness field may be structured and include the leaf-level witness components needed for a compound claim. Assessment of the tuple distinguishes checkability, adequacy, and forcing: categorical output requires adequacy and forcing, while adequacy presupposes checkability. A certificate is checkable under $\mathcal C$ when an auditor-accessible procedure $\mathrm{Check}_{\mathcal C}(\kappa)$ verifies (i) witness validity under $\langle R,t_{\mathrm{int}}\rangle$, where the certificate carries an internal witness, and (ii) the public-standing fields under $\langle S,\tau,t\rangle$, including source eligibility and the assumptions, scope, time indices, and provenance needed to place the claim in the public record.
 
Adequacy adds case-specific fit. Even a committee report that passes every certificate-level check may concern a different lease awarded by the same secretary. Such a report cannot by itself supply an adequate certificate for the lease at issue. A certificate counts as \emph{adequate} under $\mathcal{C}$ for a case $x$ only if $\mathrm{Check}_\mathcal{C}(\kappa)$ accepts, its claim and assumptions apply to $x$ (including case identity), its declared scope covers $x$ and is admissible for $\Sigma$, and $\mathcal C$ provides a contestability route. I capture this combined requirement by the adequacy condition $\mathrm{Adeq}_\mathcal{C}(\kappa,x)$, for which SA~B.1 provides a full checklist.
 
Where a certificate relies on a prose finding, authentication of the source does not verify that the finding supports the attributed claim for the identified case. The contract must specify how that attribution is verified independently of the extracting model, through a structured institutional finding or documented review by an accountable checker. Adequacy is assessed per certificate and case, subject to any conflict rules declared by $S$ (SA~C.3). Root-level conflicts are handled by the output rule below. A certificate whose production exceeds $B_{\mathrm{budget}}$ or requires
refinement beyond $N_{\mathrm{max}}$ is not adequate for $x$ under the deployed contract, whatever the intrinsic quality of its witness.
 
Forcing concerns whether the witness component settles the deployed predicate or its contractually specified negative side under $\mathcal C$, either through certified comparisons under $R$ or through the inferential rules declared by $S$ in structured evidential settings. Public standing and case applicability belong to adequacy. The relation
\[
\kappa \Vdash_{\mathcal C}\delta_\varphi(x)
\quad\text{or}\quad
\kappa \Vdash_{\mathcal C}\lnot\delta_\varphi(x)
\]
records this forcing status. A certificate is \emph{adequately forcing} for $x$ only when
\[
\mathrm{Adeq}_{\mathcal C}(\kappa,x)
\land
\kappa\Vdash_{\mathcal C}\delta_\varphi(x),
\]
or when the corresponding conjunction holds for $\lnot\delta_\varphi(x)$.
 
More concretely, the regime $R$ fixes the admissible finite witnesses, their production or propagation rules, the deployed arithmetic, and a checker $V_R$. A bound $s(x)\in[\ell,u]$ or separation witness $\ell_i>u_j$ forces the relevant comparison only under $R$'s declared precision, rounding, scaling, and input-variation assumptions, and licenses categorical output only when incorporated into an adequate certificate. Such artifacts are speech-act gates, not measures of confidence. 
 
The authority of such certificates cannot rest on unilateral system attestation. In high-stakes domains, certificates must therefore be auditor-checkable under public or jointly governed access rules. The contract must likewise preclude selective disclosure. Where $S$ prescribes checks for defeaters or contrary public grounds, an adequate certificate must include a checkable record of them, and any eligible findings they return enter the trace. Completing such checks does not show that no further contrary evidence exists. The resulting verdict remains conditional on the coverage of the declared search. If a mandatory field is withheld or unverifiable under the access rules specified in $\mathcal C$, the certificate is inadequate and cannot license $A/D$.
 
What do I require, then, when I call a certificate contestable? The contract must declare who may request re-examination, which elements -- the evidence, the decomposition, or the standard -- may be challenged, and how successful challenges revise the recorded status (SA~B.2).
 
Operationally, the Brouwerian decision layer implements the following output contract. For each input $x$, the interface returns a tuple for the warrant-status question, where $v_x \triangleq v(W_t^{S,\tau}(\varphi_x))$:
\begin{equation}\label{eq:output_contract}
\mathcal{I}_\mathcal{C}(x) = \langle v_x, \mathrm{trace}_x \rangle, \qquad v_x \in \{A,D,U\}.
\end{equation}
 
Here $\mathrm{trace}_x$ is the finite, effectively enumerable record of
certificates, checks, and reason payloads produced under $\mathcal C$. The
contract fixes $\mathrm{Adeq}_{\mathcal C}(\kappa,x)$ and $\kappa\Vdash_{\mathcal C}\chi$, where $\chi$ ranges over the deployed predicate $\delta_\varphi(x)$ and its negative side. $\kappa\in\mathrm{trace}_x$ abbreviates occurrence in the trace's certificate field. Because the trace is produced under $\mathcal C$, the quantifiers below range only over certificates exhibited within its budget and refinement horizon.
 
Formally, let
\[
A_x \equiv \exists \kappa \in \mathrm{trace}_x
\bigl(\mathrm{Adeq}_\mathcal{C}(\kappa,x)
\land \kappa \Vdash_\mathcal{C} \delta_\varphi(x)\bigr),
\]
and
\[
D_x \equiv \exists \kappa \in \mathrm{trace}_x
\bigl(\mathrm{Adeq}_\mathcal{C}(\kappa,x)
\land \kappa \Vdash_\mathcal{C} \lnot\delta_\varphi(x)\bigr).
\]
Then
\[
v_x =
\begin{cases}
A & \text{if } A_x \land \lnot D_x, \\
D & \text{if } D_x \land \lnot A_x, \\
U & \text{if } \lnot A_x \land \lnot D_x, \\
U & \text{if } A_x \land D_x \quad \text{\normalfont(\textit{U-CONFLICT})}.
\end{cases}
\]
This formalizes the uniqueness condition: the trace supports exactly one side with an adequate forcing certificate. Here, $\lnot\delta_\varphi(x)$ denotes the contractually specified negative side of the deployed predicate, not the mere failure to force $\delta_\varphi(x)$.
 
Failure of the uniqueness condition can arise from distinct sources, so the trace must report at least one interpretable reason class whenever the interface returns $U$. More than one class may apply:\footnote{The applicable class must be displayed with the $U$ status, while aggregate logs preserve class counts for detecting systematic abstention, such as persistent compute- or adequacy-limited patterns. The analogy with Brouwer's cases (3)--(4) concerns entitlement only: neither \textit{U-COMPUTE} nor \textit{U-MODEL} asserts global undecidability of the underlying world-claim.}
 
\begin{itemize}
\item \textit{U-EVIDENCE}\textbf{:} eligible public grounds are absent or insufficient to warrant either categorical move under $\langle S,\tau,t \rangle$.
\item \textit{U-ADEQUACY}\textbf{:} one or more adequacy requirements of $\mathcal C$ fail or cannot be verified, including standing, scope, provenance, case applicability, attribution, contestability, assumption fit, identity, or eligibility. The class also applies when no publicly contestable decomposition meeting the contract's adequacy criteria is available for the claim. In that case, the interface returns $U$ before applying the forcing rules.
\item \textit{U-MODEL}\textbf{:} certified bounds or margins do not force $\delta_\varphi(x)$ or its negative side under $\langle R,t_{\mathrm{int}}\rangle$.
\item \textit{U-COMPUTE}\textbf{:} the permitted evaluation exhausts its
search, refinement, or verification resources before exhibiting, within
$B_{\mathrm{budget}}$ and $N_{\mathrm{max}}$, a certificate whose witness is verified to force either side; this does not imply that no forcing witness
exists beyond the contract's limits.
\item \textit{U-CONFLICT}\textbf{:} the trace contains an adequate certificate forcing \(\delta_\varphi(x)\) and an adequate certificate forcing its negative side. Conflicts confined to a required leaf are handled under $S$ and need not reach the root (SA~C.3).
\end{itemize}
 
To support contestation, the trace must present its reason class and blocking conditions in a form usable by auditors and affected users.\footnote{For an institutional analogue, see Regulation B's requirement that adverse-action notices provide, or make available upon request, specific reasons for the decision, 12 C.F.R. \S\S~1002.9(a)(2), 1002.9(b)(2).} Categorical outputs must likewise name their decisive grounds at the surface (Section~\ref{sec:tooth_social}).
 
\subsection{Binary Totality, Witnessed Decidability, and Abstention}\label{subsec:design_lemma}
 
I now ask what a system commits itself to when it promises a certified verdict for every case in its declared scope. Read through BHK semantics, each such verdict commits the interface to $p(x)\vee\lnot p(x)$ for its deployed predicate $p$, and so requires a witness establishing or genuinely refuting it. Under the assumptions below, a system that keeps this promise therefore supplies a witnessed decision procedure on $\Sigma$. The design lemma (Proposition~\ref{prop:binary_decidability}) makes that commitment explicit. Contraposed, the result rules out certificate-sound binary totality wherever such witnessed decisions are unavailable within the contract.
 
The lemma relies on the operational assumptions fixed by the output contract: $\mathrm{trace}_x$ is finite and effectively enumerable, and the relevant adequacy and forcing checks are decidable under the declared contract $\mathcal C$. Certificate identification and the required checks are included in the declared resource budget.
 
\begin{definition}[Binary totality and certificate soundness]
Fix $\mathcal C$ and let $p$ be a deployed predicate on $\Sigma$. In applications to the warrant-status question $W_t^{S,\tau}(\varphi_x)$, $p(x)$ is instantiated by $\delta_\varphi(x)$. The results below assume that positive forcing constructively establishes $p(x)$ and that the contractually specified negative side is \emph{refutation-sound} for $p$ on $\Sigma$: whenever an adequate certificate forces that side, it constructively establishes $\lnot p(x)$. Within this subsection, the notation $\lnot p(x)$ is reserved for a negative side satisfying this condition. Contracts that instead pair $p$ with a logically independent denial predicate fall outside the scope of the proposition. For such contracts, binary totality yields witnessed two-sided resolution rather than a decision procedure for $p$.
 
An interface
\[
\mathcal{I}_{\mathcal{C}}(x)=\langle v_x,\mathrm{trace}_x\rangle
\]
is binary-total for $p$ on $\Sigma$, relative to $\mathcal C$, if for every $x\in\Sigma$ the contract-level evaluation terminates within the resources and refinement horizon fixed by $\mathcal C$ and returns
\[
v_x\in\{A,D\}.
\]
The interface is certificate-sound for $p$ on $\Sigma$, relative to $\mathcal C$, if for every $x\in\Sigma$ its categorical outputs satisfy the output contract of Section~\ref{subsec:output_contract}. Every $A$-output is backed by some adequate certificate in $\mathrm{trace}_x$ forcing $p(x)$, and every $D$-output by some adequate certificate in $\mathrm{trace}_x$ forcing the refutation-sound negative side and thereby establishing $\lnot p(x)$.
 
Finally, $p$ is \emph{witnessed-decidable} on $\Sigma$ relative to $\mathcal C$ if some terminating procedure (a \emph{witnessed decision procedure}) returns, for every $x\in\Sigma$, one side together with an adequate forcing certificate for that side, exhibited within the resources and refinement horizon fixed by $\mathcal C$. This property concerns the contract, not any particular interface trace.
\end{definition}
 
Under the specified witness rules, the principal decision predicates developed above satisfy both soundness conditions.
For the threshold predicate, the negative condition holds because a witness with $u_n(x)<\tau$ establishes $s(x)<\tau$, and hence $\lnot P_\tau(x)$. Likewise, the per-label argmax predicate is refutation-sound when the contract requires a queried-label denial witness exhibiting some $j\neq i$ with $u_n^i(x)\leq\ell_n^j(x)$, thereby establishing $\lnot M_i(x)$.
 
Under both soundness assumptions, adequately forcing certificates for $p(x)$ and its negative side cannot coexist in a consistent regime. A conflict trace then signals failed soundness or compatibility assumptions, and \textit{U-CONFLICT} is a protective output. In standing-relative regimes, by contrast, rival eligible warrants may coexist without establishing contradictory propositions. Where either soundness condition is not satisfied, the proposition's hypotheses are not met. I therefore use the \textit{Tooth Social} case only to illustrate the output contract under defeasible public standing. It does not assume that the corruption predicate satisfies the forcing-soundness conditions of the design lemma.
 
\begin{proposition}[Design lemma: contract-respecting binary interfaces imply witnessed decidability]\label{prop:binary_decidability}
Any binary-total, certificate-sound interface for $p$ on $\Sigma$, relative to $\mathcal{C}$, yields a witnessed decision procedure for $p$ on $\Sigma$.
\end{proposition}
 
\begin{proof}
Given $y\in\Sigma$, run the contract-level evaluation of $\mathcal I_{\mathcal C}(y)$. By binary totality, it terminates within the declared resources and refinement horizon and returns $v_y\in\{A,D\}$. By certificate soundness, $\mathrm{trace}_y$ then contains an adequate certificate forcing the selected side; for $v_y=D$, refutation soundness ensures that this certificate constructively establishes $\lnot p(y)$. Since $\mathrm{trace}_y$ is finite and effectively enumerable, the adequacy and forcing checks are decidable under $\mathcal C$, and certificate identification falls within the declared resource budget, the induced procedure $D_{\mathcal C}$ can search the trace, select the first such certificate, and return the corresponding side together with it. Hence $D_{\mathcal C}$ terminates on every $y\in\Sigma$ with one side and an adequate forcing certificate for that side, and is therefore a witnessed decision procedure for $p$ on $\Sigma$, relative to $\mathcal C$.
\end{proof}
 
\begin{corollary}[Conditional incompatibility]\label{cor:incompatibility}
Fix $\mathcal C$ and a deployed predicate $p$ on $\Sigma$, under the effectiveness and forcing-soundness assumptions of Proposition~\ref{prop:binary_decidability}. If $p$ is not witnessed-decidable on $\Sigma$ relative to $\mathcal C$, then no interface for $p$ on $\Sigma$ is both binary-total and certificate-sound relative to $\mathcal C$.
\end{corollary}
 
\begin{proof}
Immediate by contraposition of Proposition~\ref{prop:binary_decidability}.
\end{proof}
 
\begin{corollary}[Mandatory abstention without a uniquely licensed side]
Fix $\mathcal C$ and let $x\in\Sigma$, with $A_x$ and $D_x$ defined as in Section~\ref{subsec:output_contract}. If the permitted trace supports neither side,
\[
\lnot A_x \land \lnot D_x,
\]
then no certificate-sound interface may return $A$ or $D$ at $x$. If the interface is total over $\{A,D,U\}$, it must return $U$ with the applicable reason class. If the permitted trace supports both sides,
\[
A_x \land D_x,
\]
the interface must likewise return $U$, with reason class \textit{U-CONFLICT}.
\end{corollary}
 
\begin{proof}
Immediate from certificate soundness, totality over $\{A,D,U\}$, and the output rule in Section~\ref{subsec:output_contract}.
\end{proof}
 
The proposition's significance lies less in its proof than in the design commitments its contrapositive makes incompatible.
 
\medskip
\textbf{Practical consequence.} As noted in Section~\ref{subsubsec:forcing_thresholds}, even a decidable floating-point comparison need not license categorical output when $\mathcal C$ demands additional guarantees concerning input variation or rounding. The idealized constructive-real semantics exhibits a related constraint. Real-number trichotomy ($a<b$, $a=b$, or $a>b$) is not constructively available in general, and categorical comparison requires an exhibited witness settling the relevant side \citep[pp.~22--23]{bishop1985constructive}. Unique-argmax claims likewise require constructively witnessed strict inequalities, which call for robust separation margins or weakened certified predicates, such as certified top-$k$ sets. Hence, unless $\Sigma$ is restricted to cases in which such structure is guaranteed, an interface that is both certificate-sound and total cannot in general remain purely binary.
 
\medskip
\textbf{Design alternatives.} Given the incompatibility, an interface may issue uncertified verdicts, abandoning certificate soundness; treat failure to establish $p$ as grounds for denial, abandoning refutation soundness; or narrow its declared scope to exclude cases that remain unsettled. The first two weaken the constructive discipline, while the third changes the contract and must be publicly accounted for. I defend a fourth alternative, which preserves the declared scope and relinquishes binary totality: where the permitted trace fails to license exactly one side, the interface returns $U$ and owes those affected an account of what keeps the case open. These obligations follow from the output contract, not from the design lemma alone.
 
This brings the deployment contract into view as a governance object. Since witnessed decidability is always relative to $\Sigma$, $R$, $B_{\mathrm{budget}}$, $N_{\mathrm{max}}$, $S$, $\tau$, and $t$, abstention rates cannot be interpreted independently of those parameters. A low rate of $U$ is therefore ambiguous. Such a rate may reflect genuine certificate coverage, a narrowed scope, a weakened evidential standard, a shifted threshold, or a changed witness discipline. Suppose two deployments of \textit{Tooth Social} report the same fall in abstentions: one has obtained records that settle previously open cases, the other has declared out of scope all inputs requiring congressional findings. Only the first has extended its certified coverage. Out-of-scope high-stakes queries therefore receive a scope notice, not a categorical verdict, and exclusions are logged separately from $U$-rates. For high-stakes deployments, the contract's parameters must accordingly be subject to external auditability and change control. Otherwise, abstention rates can be manipulated toward either artificial decisiveness or responsibility-laundering.
 
\medskip
\textbf{Interface-level analogue.}
I turn to debates about epistemic conceptions of truth because they show how the absence of warrant for a claim can come to look like warrant for its denial. Hilary Putnam's proposal \citep[p.~55]{putnam1981reason} to treat truth as idealized rational acceptability faces Crispin Wright's objection \citep[Chapter~2]{wright1992truth} that combining an evidential constraint on truth with minimal truth platitudes and a classical treatment of negation pushes one toward an excluded-middle-style completeness claim. Epistemically ideal conditions cannot be neutral between $P$ and $\lnot P$. Failure to warrant $P$ would therefore entail warrant for $\lnot P$, so genuinely undecidable cases disappear. Wright recommends a broadly \textit{intuitionistic} revision that rejects excluded middle for arbitrary $P$ in discourses where there is no guarantee of decidability. Proposition~\ref{prop:binary_decidability} recasts this tension.\footnote{The parallel is interface-level rather than regime-specific. In interval regimes, the gap appears as a failure of separation. In public evidential regimes, it appears as a failure of adequate standing or record warrant.} Insisting on both certificate soundness and total binary outputs for a deployed predicate commits a system to witnessed decidability on its declared scope. Read this way, the status $U$ is the interface-level counterpart of intuitionistic resistance to unwarranted completeness.
 
\medskip
\textbf{Design upshot.} A binary verdict issued when the uniqueness condition fails does more than risk error. It presents insufficient or conflicting warrant as settled entitlement. New grounds, revised standing, or an explicit and legitimate change in the contract may license a later verdict, but only if the resulting trace contains an adequate forcing certificate for exactly one side. Suppressing $U$ changes the output, not its warrant.
 
\section{Limitations}\label{sec:limitations}
 
I offer the Brouwerian constraint as a normative design principle. It applies most directly where AI outputs enter public justificatory practices -- including law, public policy, and democratic deliberation. Such a standard is demanding but bounded. It requires the exhibited warrant needed for the speech act, not exhaustive certainty or complete institutional settlement.

Implementation may impose substantial verification costs, especially in hard real-time settings. The constraint does not, however, require fully constructive neural architectures, since conventional models can operate behind certificate-gated interfaces. Where adequate certificates cannot be exhibited within the declared resources, the system must return \textit{Undetermined}, potentially reducing categorical coverage.
 
Certificate gating also presupposes that the system knows which claim it is making. I do not solve the problem of individuating claims during open-ended generation. Categorical output requires prior claim routing, a contractually specified and checkable mapping from the proposed utterance to a claim $\varphi$ and deployed predicate $\delta_\varphi$. Where that routing remains unresolved, the failure concerns adequacy, and the interface returns $U$ for that utterance with reason class \textit{U-ADEQUACY} before the forcing rules are applied. Near-term implementations can sidestep the problem by constraining outputs through schemas, templates, or tools that fix claim identity by construction.
 
Claim routing sets the reach of the gate. AI persuasion extends beyond assertion to salience, framing, selection, and affective tone \citep{floridi2024hypersuasion}. My focus here is the narrower case of fluent prose that presents contested world-claims as settled. The direct guarantee covers certificate-gated categorical outputs. Pragmatically conveyed verdicts, such as insinuations, fall within that discipline only insofar as claim routing identifies the commitments they introduce. Purely affective or selection-based influence lies beyond the gate and calls for complementary measures such as privacy, accountability, competition, and education.
 
The remaining three limitations are institutional and concern the authority of certificates, the justification of thresholds, and the governance of abstention.
 
\textit{First}, an adequate forcing certificate establishes compliance with the declared mathematical bounds and the specified public record. It leaves open the truth of the underlying claim, the impartiality of the evidence or training pipeline, the legitimacy of public standards, and the substantive adequacy of the model's representations.
 
\textit{Second}, threshold specifications such as $\tau$ cannot be fixed by mathematics alone. They are normative and institutional parameters to be set through legal, democratic, or domain-specific processes. By making such thresholds explicit and auditable, the framework resists technocratic gatekeeping and opens threshold-setting to public contestation.
 
\textit{Third}, abstention is not neutral. $U$ is also a behavioral, organizational, strategic, and distributive outcome. Institutions must therefore specify how $U$ triggers escalation, delay, or human review so that abstention does not become a loophole for evasion or bureaucratic deferral. Suppose a benefits office must decide a citizen's claim by a statutory deadline while the interface returns $U$. Governance then has to distinguish warrant status from action policy. A reason-coded $U$ does not itself determine whether the institution should delay, escalate, follow a presumption, or adopt a precautionary default. In criminal proceedings, the presumption of innocence likewise governs treatment independently of the interface's status.
 
Claims about named persons raise a related difficulty. A persistent $U$ can itself suggest wrongdoing. Since constructively denying an open allegation requires excluding every admissible witness route (Section~\ref{sec:tooth_social}), $U$ may persist for structural reasons without indicating growing suspicion. Reports should explain this, not merely disclaim any adverse inference. I therefore propose that contracts for such claims deploy specified allegations, which name a particular act and transfer, alongside the open one. Refuting a necessary conjunct then defeats the specified allegation without excluding other admissible witness routes. Negative-warrant clauses must also distinguish acquittal, insufficient evidence, and factual innocence. Record-based denial, including by scalar upper bound, requires $\mathcal C$ to specify how potentially missing evidence is accounted for.
 
Abstention has distributive effects \citep[see also][p.~9]{wu2026capability}. Groups that are underrepresented or systematically mismeasured may receive disproportionate $U$ outcomes. Where this produces delay, heightened scrutiny, or effective exclusion, a certificate-gated interface may reproduce or exacerbate existing injustice. $U$-rates should accordingly be tracked by subgroup, linked to procedural safeguards, and paired with guaranteed pathways for timely escalation and contestation.\footnote{High-stakes deployments should also log the provenance of the gap: world-grounded, pipeline-grounded, or conflict-grounded (SA~B.3).}
 
Taken together, these limitations reveal a central institutional dependency of the framework. The Brouwerian constraint can discipline when a system may speak categorically, but it cannot legitimate the public standards, thresholds, decompositions, or adjudicative procedures on which such speech rests \citep[see also][pp.~4, 9]{wu2026capability}. Conflict cases show this most clearly. A conflict-grounded $U$ can disclose strategically induced uncertainty and trigger escalation to substantive public adjudication. Resolving the conflict within the warrant-relevant time window remains the task of that adjudication. I regard this dependency as structural. Once machine speech is made answerable to public warrant while the system disclaims authority to constitute that warrant, disputes over standards, unresolved conflicts, and challenges necessarily return to the responsible authorities.
 
The same dependency runs through \textit{Tooth Social}, where the citizen accepted the committee's finding. Had the committee been captured, that finding could still have qualified under the declared standard, since eligibility does not by itself establish legitimacy. A facially adequate finding from such an institution would launder compromised evidential authority into contractually licensed assertion. A captured committee could also selectively release subpoenaed records, distorting receipt and valuation evidence. The contestability route makes the reliance on that finding visible and open to challenge. It neither detects institutional capture on its own nor confers independent legitimacy on the finding (SA~C.6). Moreover, challenging it presupposes the very public epistemic capacities whose erosion motivates the constraint.

\section{Conclusion}\label{sec:conclusion}
 
The Brouwerian assertibility constraint I have defended does not settle contested world-claims. In high-stakes domains, it limits the authority with which machines may speak about them. Under the output contract, \textit{Undetermined} is mandatory unless the permitted trace supports exactly one side with an adequate forcing certificate. The design lemma shows what eliminating that status would cost. Given its stated assumptions, a binary-total, certificate-sound interface yields a witnessed decision procedure on its declared scope. Where no such procedure is available, the two properties cannot be combined. The interface thus serves as a \textit{speech-act gate}, directing users from fluent verdicts toward reasons, records, and accountable procedures. In this limited way, it helps resist the deskilling that threatens epistemic agency \citep{ferdman2025ai}.
 
By pairing \textit{Undetermined} with reason classes and an indication of what would count as grounds for moving to licensed assertion or licensed denial, the framework keeps high-stakes exchanges open to acknowledgment. It also makes the decomposition and governing standard explicit, which lets citizens challenge both the evidence offered and the terms under which it counts as warrant. Neither step guarantees successful contestation or legitimates the institutions defining those terms.
 
My proposal does not attempt to give models an inner truth principle. Instead, it draws a normative boundary between fluent output and warranted categorical speech. A witness, as the term is used here, is no substitute for lived experience and performs no acknowledgment. It provides a checkable ground on which categorical machine speech can be inspected, challenged, and refused. Read in this way, the requirement that systems either exhibit their warrant or return \textit{Undetermined} serves a Cavellian concern with acknowledgment and answerability. If models are to participate as intermediaries in democratic life, they must do so under norms that keep their fluency from displacing the work of finding and standing behind our own words.
 
A fitting place to end is Brouwer's insistence that intuitionistic mathematics concerns ``inner architecture,'' and that foundational inquiry can have ``revealing and liberating consequences, also in non-mathematical domains of thought'' \citep[p.~1249]{brouwer1949consciousness}. In our context, the parallel concerns the architectures through which AI systems now mediate democratic speech. For an autoregressive language model, categorical output is licensed only if the claim expressed has been checkably identified and supported by an adequate certificate. If such architectures amount merely to a ``meaningless building of words,'' they risk assuming the role Brouwer feared in formalism: blind ``regulators of social struggle'' that impose ``norms of mechanization'' on human beings \citep[Brouwer, as cited in][p.~ix]{vandalen2026intuitionistic}. 
 
To avert that prospect, I have argued, architectures must be disciplined by certificate-based norms of assertion, and citizens and institutions must be able to understand, contest, and revise the grounds and standards behind categorical outputs. In Stanley Cavell's terms, happiness is won not only by opposing those in power but also by ``educating them, or their successors'' \citep[p.~4]{cavell1981pursuits}. Applied to our ``Algorithmic City,'' that education requires making automated speech answerable to shared records and checkable witnesses. This extends Cavell's education of grownups \citep[p.~125]{cavell1979claim} to our shared technological life, where learning to detect ``rot'' includes questioning the standards by which machine speech acquires public authority. Through democratic practices, we must learn to hold our machines to account and, in doing so, continue our own education as epistemic agents.

\end{document}